%% file: paper.tex
\documentclass[11pt]{article}

\usepackage{fullpage}
\usepackage{graphicx}
\usepackage{amsmath,amsthm,amsfonts}
\usepackage{xcolor}
\usepackage[hidelinks]{hyperref}
\usepackage{framed}
\usepackage{enumitem}
\usepackage{xcolor}

\usepackage{subcaption}

\newtheorem{theorem}{Theorem}[section]
\newtheorem{lemma}[theorem]{Lemma}
\newtheorem{corollary}[theorem]{Corollary}
\newtheorem{claim}[theorem]{Claim}

\theoremstyle{definition}

\newtheorem{definition}[theorem]{Definition}
\newtheorem{question}[theorem]{Question}
\newtheorem{remark}[theorem]{Remark}

\newcommand{\set}[1]{\{#1\}}

\newcommand{\eps}{\epsilon}
\DeclareMathOperator{\poly}{poly}
\DeclareMathOperator{\triangles}{triangles}
\newcommand{\eqdef}{\stackrel{\text{\tiny\rm def}}{=}}

\newcommand{\alg}{{\mathcal A}}

\newcommand{\fv}[1]{v^{(#1)}} 
\newcommand{\ft}[1]{F_2(\fv{#1})} 

\newcommand{\BMAM}{Bounded--Memory Adversary Model}
\newcommand{\tSAM}{$\tau$--Stream Adversary Model}

\input{macros}

\title{
Robust Streaming Against Low--Memory Adversaries %
}
\author{
Omri Ben-Eliezer\\
Technion
\and
Krzysztof Onak\\
Boston University
\and
Sandeep Silwal\\
University of Wisconsin-Madison}
\date{November 2025}

\begin{document}
\maketitle

\begin{abstract}
Robust streaming, the study of streaming algorithms that provably work when the stream is generated by an adaptive adversary, has seen tremendous progress in recent years. However, fundamental barriers remain: the best known algorithm for turnstile $F_p$--estimation in the robust streaming setting is exponentially worse than in the oblivious setting, and closing this gap seems difficult.
Arguably, one possible cause of this barrier is the adversarial model, which may be too strong: unlike the \emph{space--bounded} streaming algorithm, the adversary can memorize the entire history of the interaction with the algorithm. Can we then close the exponential gap if we insist that the adversary itself is an \emph{adaptive but low--memory} entity, roughly as powerful as (or even weaker than) the algorithm?

In this work we present the first set of models and results aimed towards this
question. We design efficient robust streaming algorithms against adversaries
that are fully adaptive but have no long--term memory (``memoryless'') or very
little memory of the history of interaction. Roughly speaking, a memoryless
adversary only sees, at any given round, the last output of the algorithm (and
does not even know the current time) and can generate an unlimited number of
independent coin tosses. A low--memory adversary is similar, but maintains an
additional small buffer. While these adversaries may seem quite limited at first
glance, we show that this adversarial model is strong enough to produce streams
that have high flip number and density in the context of $F_2$--estimation,
which rules out most of known robustification techniques. We then design a new
simple approach, similar to the computation paths framework, to obtain efficient
algorithms against memoryless and low--memory adversaries for a wide class of
order--invariant problems. We conclude by posing various open questions
proposing further exploration of the landscape of robust streaming against fully
adaptive but computationally constrained adversaries.
\end{abstract}

\noindent

\section{Introduction}
Adversarially robust streaming, i.e., the study of streaming algorithms that are robust against adaptive adversaries, has received a great deal of attention in the last few years (see, e.g., \cite{hardt_woodruff, BY20,advrob, dp_advrob, WZ2021, braverman2021adversarial, BEO22,  cs_one, cs_two,SSS23,cherapanamjeri2023robust,GribelyukLWYZ24,WZ2024,GribelyukLWYZ25,MN2021,menuhin2025shufflingcardslittlebrain,ACSS24,KMNS21,adaptive_cms} and the many references within). In the streaming setting, the goal of the algorithm is to approximately compute some function over the stream---say, the $F_p$--moment\footnote{Assuming that the elements in the stream come from a finite universe, such as $[n]$, the \emph{$F_p$--moment} of the stream is $\sum_{i \in [n]} |f_i|^p$, where $f_i$ is the frequency (possibly negative) of item $i$ in the stream.}---using as little space as possible, ideally using memory that is polylogarithmic in the stream length $m$ and the underlying universe size $n$, and polynomial in the inverse of the approximation parameter $\eps$. In the standard (black box) adversarially robust streaming model, the stream is generated by an adversary who sees and memorizes all outputs of the algorithm along the stream. The adversary can choose each stream update in a way that depends on the \emph{entire history} of the interaction between the adversary and the algorithm---that is, on all stream updates and algorithm outputs so far. This is in contrast to the standard oblivious model, in which the stream updates do not depend on the actions of the streaming algorithm (i.e., the stream is thought of as generated beforehand and is not adaptive).

Arguably the most interesting open question in adversarially robust streaming
concerns the space complexity of turnstile moment estimation
($F_p$--estimation), in particular for $p=0$ and $p=2$ \cite{Jayaram2021}. In
the turnstile streaming setting, both insertions and deletions of elements are
allowed. For simplicity, we shall consider the setting of unit updates, in which each stream update
either inserts or deletes a single copy of one item chosen by the adversary (and
thus changes the frequency of the item by $+1$ or $-1$,
respectively).
Unfortunately, there are exponential gaps between the optimal space complexity against oblivious adversaries and the best--known streaming algorithms in the presence of adaptive adversaries. In the oblivious case, algorithms based on linear sketching attain space complexity that is logarithmic in $m$ and $n$ for $F_p$--estimation when $p \leq 2$. Meanwhile, in the adversarially robust regime, all existing solutions have space complexity polynomial in the length $m$ of the stream---an exponentially worse dependence in $m$. For example in $F_2$--estimation, the best known algorithms \cite{BEO22,WZ2024} require space complexity $\tilde{O}(m^{0.4})$, which is exponentially worse than that of linear sketches in oblivious streams. Thus, closing this gap between oblivious and robust turnstile $F_p$--streaming is a central and likely challenging open problem \cite{Jayaram2021}.

Upon a closer inspection of the robust streaming setting, however, one can easily see that there is an inherent power imbalance between 
the algorithm and the adversary in this model. While the streaming algorithm is required to use a small amount of space (much smaller than $m$, the length of the stream), the adversary inherently has memory of size at least $m$, as it retains all previous stream updates and all previous outputs of the algorithm.

In this work we set out to explore streaming models that maintain a better balance between the respective computational power (in this case, the amount of memory) of the algorithm and that of the adversary. Our main high--level motivating question is:
\begin{quote}
\centering
    \emph{Is it possible to obtain robust streaming algorithms for, say, $F_2$--estimation\\against adversaries with memory size $k$, where the space complexity of the\\algorithm is polylogarithmic in $m$ and $n$, and polynomial in $1/\eps$ and $k$?}
\end{quote}
While we are not able to fully resolve the above question (and it remains a very intriguing open problem), we obtain the first positive results of this flavor.

Importantly, there are other beyond worst case models in the literature that
constrain the adversary. Examples include models in which the adversary has a
limited amount of adaptive interaction with the algorithm (i.e., it can modify
the future input stream only a limited number of times)
\cite{SSS23,cherapanamjeri2023robust}; or models in which the algorithm
occasionally receives a small amount of advice from a powerful (and not
memory--constrained) oracle \cite{SSS23}. However, none of these models capture
the natural setting considered in this paper, in which the adversary is fully
adaptive---that is, it sees all outputs of the algorithm and can immediately
edit the next input after each algorithm's output---but constrained in terms of
its internal memory size.

\subsection{Our Model}
Let us first quickly recall the standard adversarial model \cite{BY20,advrob}.
At a high level, the streaming algorithm receives $m$ updates of items
(insertions and deletions) from a finite universe of size $n$. In general, $n$
can be much bigger than $m$, but without loss of generality, one can assume that
$n = O(\poly(m))$, because one can randomly hash the stream elements into $O(\poly(m))$
buckets, while ensuring no collisions for different items with high probability
(with the caveat that one would also need space to store the hash function).
Thus, for simplicity, we assume items are elements in the domain
$[n]$. A key concept is the notion of a \emph{frequency vector} of the
underlying stream, which simply records the number of times every item has
appeared (where a deletion is counted as a negative occurrence, see Definition
\ref{def:streams} for details).

The goal of the streaming algorithm is to estimate functions of the underlying
stream, where at any point in the stream, the function only depends on the
current frequency vector. We consider the \emph{tracking} variant in which the algorithm
must output an estimate after every stream update.  Perfect estimation is of
course possible by exactly recording the frequency vector, but the goal of the
entire field of streaming algorithms is to design algorithms
that operate in space sublinear in the size of the input
(in our case, ideally poly-logarithmic in $n$ and $m$).
This setting is quite
general and captures a wide range of streaming problems, including
$F_p$--estimation and graph streaming, among others
\cite{Muthu2005,McGregor2014}.

\paragraph{The bounded--memory adversarial model.}
In our adversarial setting,
the stream updates are controlled by a possibly randomized adversary. Crucially, the adversary is memory--bounded and its actions only depend on the \emph{most recent} estimate of the algorithm, with possibly a very small amount of longer--term memory. To model this, we define several types of memory for the adversary. See Figure~\ref{fig:model_diagram} for a visual depiction and Section \ref{sec:model_definition} for a formal definition of the model.

\begin{itemize}
\item \textbf{Estimate memory:} We assume the adversary has an \emph{estimate memory} where the most recent estimate of the algorithm (and nothing else) is always written. (We do not bound the amount of this memory but for problems that we consider, a good approximation to a function's value can easily be provided in small memory, using the standard number representation.)
\item \textbf{Working memory:} We equip the adversary with a \emph{working memory}, which is also unbounded. This allows the adversary to compute an arbitrary function of the latest estimate in the estimate memory.
Note that this memory would easily allow the adversary to write down and remember the entire history of  the algorithm's outputs, thus reducing our setting to the general robust streaming one. Thus, in our setting, the working memory is wiped clean after the adversary computes the next update for the streaming algorithm.
\item \textbf{Persistent memory:} The adversary may have a (limited amount of) \emph{persistent memory} which \emph{does} allow the adversary to store and update a bounded amount of information throughout the entire stream. If the persistent memory is limited to zero bits, then we refer to the adversary as \emph{memoryless}.
\item \textbf{Randomness:} Lastly, we also allow the adversary to use as much randomness as it desires, and assume it has access to random bits freshly drawn at every round. These freshly generated random bits are not retained for future rounds.\footnote{ One could also consider allowing the adversary to generate and memorize an unlimited amount of randomness in the beginning of the process, before any interaction with the adversary. However, as we discuss in Remark~\ref{remark:randomness}, this would not effectively change the power of the adversary.}
Thus, a memoryless algorithm samples an update from some distribution $P_y$ that depends on the latest
estimate $y$ provided by the algorithm.
For a low--memory algorithm, this distribution depends on the last output and on the current state of the persistent memory. As usual, a memoryless (or bounded--memory) adversary is deterministic if it does not use any random bits during its operation.
\end{itemize}

\begin{figure}[h]
\hspace*{-1.5cm}  
    \includegraphics[width=20cm]{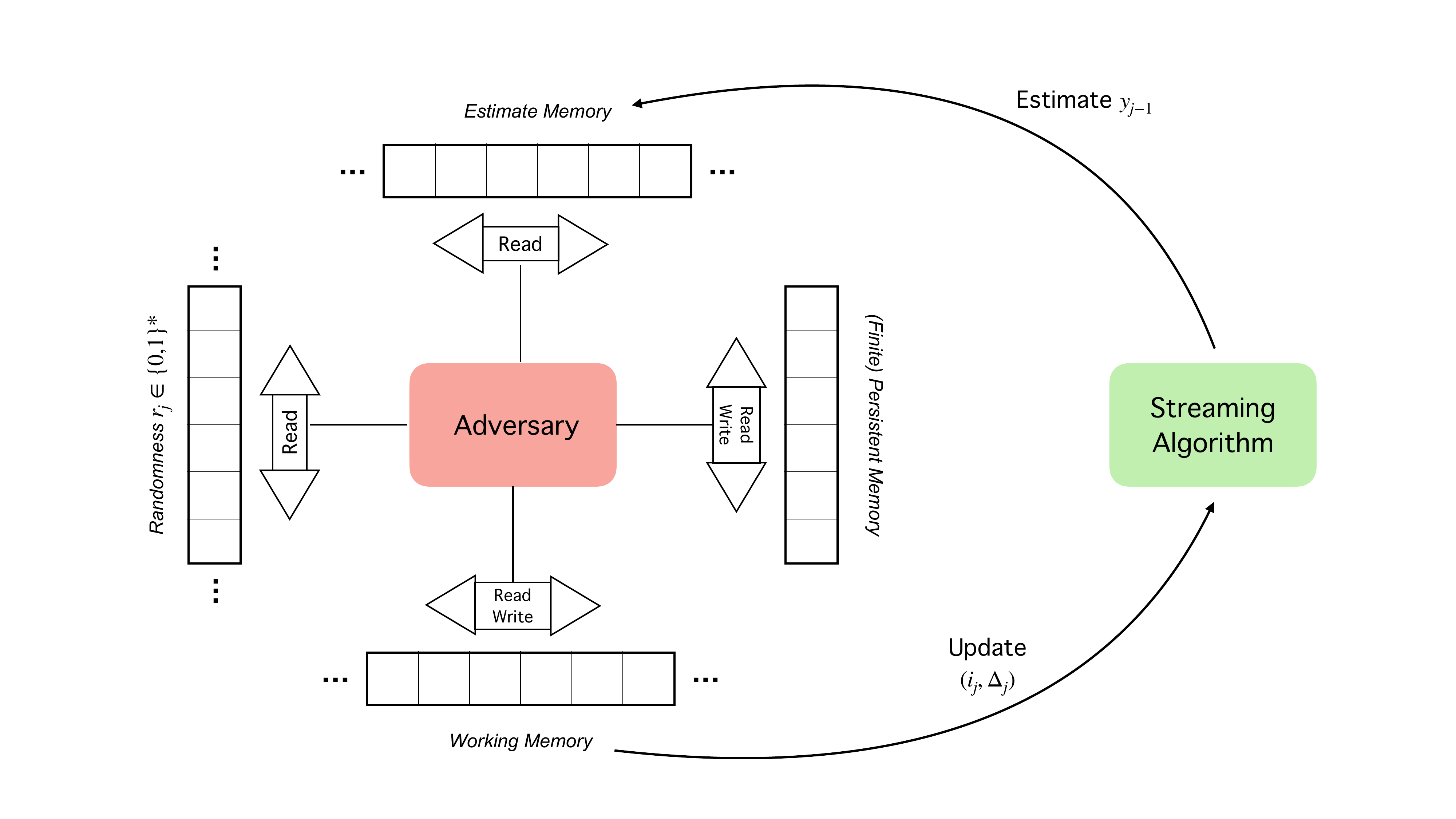}
    \caption{A pictorial representation of our streaming model. Not shown: The estimate memory is overwritten by the estimate $y_{j-1}$ of the streaming algorithm, and the working memory is wiped clean after each update $(i_j, \Delta_j)$ is computed. We also don't show the inner mechanisms of the streaming algorithm or the potential strategy the adversary is following. See Definition \ref{def:robust_streaming} for the formal presentation.}
    \label{fig:model_diagram}
\end{figure}

\subsection{Our Results}\label{sec:our_results}

\paragraph{Robust streaming against deterministic memoryless adversaries.}
The weakest fully--adaptive adversary whom one could imagine in our context is a
deterministic memoryless adversary. Note that for this adversary, the random
read--only memory and the persistent memory are empty, and the working memory
starts each round empty as well. So the adversary makes the decision what update
to send only based on the last output of the algorithm.  Through the use of a
standard rounding--based mechanism and sparsity arguments, we obtain the
following.

\begin{theorem}[Deterministic adversaries; informal version of Theorem \ref{thm:deterministic_adversary} and Corollary~\ref{thm:deterministic_adversary2}]
\label{thm:det_informal}
There exists a deterministic streaming algorithm that, for a fixed function $f$ of the underlying stream frequency vector
with range $\{0\} \cup [1,\alpha]$,
tracks the value of $f$ up to a multiplicative factor of $1+\eps$
against a deterministic memoryless adversary, using $O\left(\frac{\log(\alpha)}{\eps} \cdot \log(mn)\right)$ bits of space. If the adversary is deterministic and has $k$ bits of persistent memory, the space bound is $O\left(2^k \cdot \frac{\log(\alpha)}{\eps} \cdot \log(mn)\right).$
\end{theorem}
We note that the above theorem poses almost no constraints on the function $f$---that is, essentially \emph{any} reasonable function of the frequency vector (beyond just specific special cases, such as $F_p$--moments) can be efficiently and deterministically approximated against a deterministic adversary. These results are proved in Section~\ref{sec:deterministic_adv}.

\paragraph{Memoryless adversaries can generate high flip--number, high density streams.}
The next natural question is whether it is possible to design adversarially robust algorithms against \emph{randomized} low--memory adversaries. Before proving the positive results on this front, we demonstrate that the class of streams generated by these adversaries is sufficiently interesting and non-trivial to resist a black--box application of existing robust streaming frameworks in the literature. Specifically, we prove that such an adversary can generate streams that have high \emph{density}\footnote{The density of the stream (or more precisely, of the frequency vector of the stream) at any point during the process is the number of non-zero entries in the frequency vector at that time. See Definition \ref{def:sparsity}.} and large \emph{flip number},\footnote{The flip number is roughly the number of times the underlying value of $f$ over the stream changes by a constant multiplicative factor, see Definition \ref{def:flip_number}.} in the case where $f$ is $F_2$, the second moment of the frequency vector.

\begin{theorem}[Low--memory stream with high density and large flip number; informal version of Theorem~\ref{thm:lb_construction}]
Fix any constant $0 < c < 1$.
There exists a randomized memoryless adversary for $F_2$--estimation that, conditioned on the algorithm always outputting a $(1+\eps)$--approximation in the specific instance of the adversarial game, generates a stream of length $m$ whose density is $\Omega(m^c)$, for all but the first $m^c$ rounds of the stream, and whose flip number is $\Omega(m^{1-c})$.

In addition, if we provide the adversary with a single bit of persistent memory, then the above statement holds, for any constant $0 < c < 1$, with flip number $\Omega(m^{1-0.5 c})$.
\end{theorem}

Again the result here is rather general in a certain interesting sense. While
the result applies only to the special case of $f = F_2$, it only concerns the
adversary and does not assume anything specific about the behavior of the
algorithm, aside from requiring that it provides good estimates
throughout the stream.
In particular, it also applies to algorithms that can memorize the entire stream.

We note that it is easy to construct a stream with flip number $\Omega(m)$,
but density $O(1)$, using a deterministic memoryless adversary: simply take an
adversary who when the algorithm outputs zero, inserts the element, say, $1$;
and when the algorithm outputs any value bigger than zero, the adversary deletes
the same element. It is also easy to construct a stream with flip number of
order $\Theta(\log m)$ and density $\Omega(m)$, this time with a randomized memoryless
adversary (provided the underlying universe is large enough): simply insert a
random element from the universe in each round. However, each of these examples
in itself is not very interesting and can be easily handled by at least one of
the existing generic black--box frameworks for robust streaming. But as it turns
out, these two simple constructions can be interwoven together (by a randomized
memoryless adversary) to produce a stream with both high (polynomial in $m$)
flip number and high density.

The state of the art robust streaming algorithms for $F_2$--estimation against
\emph{any} black--box adversary (not just bounded adversaries as in this paper)
have space complexity $\tilde{\Theta}(m^{0.4})$ and are based on dense--sparse
tradeoffs. The earlier and simpler of them, by Ben-Eliezer, Eden, and Onak
\cite{BEO22}, has space complexity that scales with the maximum between the
sparsity and the square root of the flip number for the non-sparse part of the
stream. The more recent one, due to Woodruff and Zhou~\cite{WZ2024}, is more
intricate, attains the same space complexity for $F_2$--estimation (and
improved complexity for $F_p$ with $p \in(1,2)$) but is better tailored for
streams in which the dense part involves insertions and deletions on a small set
of heavy hitters. Informally speaking, our construction toward the proof of
Theorem~\ref{thm:lb_construction} appears to pose a barrier for the algorithm of
Ben-Eliezer et al.\ (and more generally any algorithm whose complexity depends
as a black box on the sparsity of the stream and the flip number of the dense
part), but the Woodruff--Zhou algorithm appears better tailored for processing
such a stream efficiently since the dense part in our construction involves
repeated insertions and deletions of the same single element. Nevertheless, the
fact that memoryless adversaries can generate streams with both large flip
number and large density provides a strong indication that such a model, albeit
seemingly weak at first glance, is interesting and strong enough to fool many of
the existing approaches in the literature. We leave open the question of whether
the Woodruff--Zhou algorithm (or variants thereof) can solve $F_2$--estimation
against memoryless adversaries with space complexity polylogarithmic in $m$ (see
Section \ref{sec:open_problems}). Even if this is the case, our algorithms
have the advantage of being significantly simpler.

\paragraph{Handling memory--bounded randomized adversaries.}
We now go beyond the deterministic adversary case and turn our attention to
randomized adversaries. Our main theorem is the following.

\begin{theorem}[Randomized adversaries; informal version of Theorem
\ref{thm:robust_low_memory}]
Suppose there exists an order--independent oblivious streaming algorithm that
uses $M(\eps)$ space to compute a $(1+\epsilon)$--approximation
to a function $f$ with range $\set{0} \cup [1,\alpha]$
with probability at least $9/10$.
Then, there exists a corresponding robust algorithm
with space $M(\epsilon/3) \cdot O\left(\frac{\log
\alpha}{\eps}\cdot \log m
\right)$ that tracks the value of $f$
up to a factor of $1+\epsilon$ against a memoryless randomized adversary,
with probability at least $9/10$ over the entire stream.
If the adversary has $k$ bits of persistent memory,
the space bound increases to $M(\epsilon/3) \cdot O\left(2^k \cdot \frac{\log
\alpha}{\eps}\cdot \log m
\right)$.
\end{theorem}

To prove the above theorem, we first introduce a new adversarial model, which we
call the $\tau$--Stream Adversary Model, formally introduced in
Definition~\ref{def:k_stream_adversary}.
This adversarial model, parameterized by an integer
$\tau \ge 1$, is very much similar to our main bouded--memory model (see Definition \ref{def:robust_streaming}),
but we allow for it to have unbounded
persistent memory, and it can thus memorize the entire interaction history with
the streaming algorithm.
To counter its power, we only allow it to
switch (arbitrarily) between $\tau$ streams of pre-generated updates as it
interacts with the streaming algorithm. For this adversary, we show how to turn
a correct order--independent oblivious streaming algorithm into one which can
handle the $\tau$--Stream Adversary Model at the cost of, roughly, a \emph{linear} in
$\tau$ blow-up in space. At a high level, the order--independent property allows
us to union bound over the selection of prefixes of the initial $\tau$ streams,
of which there are only $O(m^{\tau})$. Using the standard median of estimates trick,
this implies that our overall space dependence is logarithmic in the number of options,
resulting in the claimed linear in $\tau$ blow-up. Note that this is reminiscent
of the computational paths technique \cite{advrob}. See  Lemma
\ref{lem:k_stream} for a formal argument.

This intermediate $\tau$--Stream Adversary Model is useful because we can reduce the
randomized adversary in the bounded--memory model to a $\tau$--Stream Adversary
for an appropriate value of $\tau$. More formally, the proof of Theorem
\ref{thm:robust_low_memory} gives a reduction from the randomized adversary
with $k$ bits of persistent memory to the $\tau$--Stream Adversary Model for
$\tau = O(2^k \eps^{-1} \log \alpha)$. At a high level, this comes from the fact
that the persistent memory can take on at most $2^k$ different states and the
total number of ``different'' estimates of the streaming algorithm can be assumed
to be $O(\eps^{-1} \log \alpha)$, via a rounding argument (see Lemma
\ref{lem:net}). Hence, we can view the adversary
with $k$ bits of persistent memory as interleaving between
$\tau = O(2^k \eps^{-1} \log \alpha)$ many pre-defined streams. See the proof of
Theorem \ref{thm:robust_low_memory} for details. Lastly, we remark that Theorem
\ref{thm:robust_low_memory} has downstream corollaries to order--independent
problems, including moment estimation. See Section~\ref{sec:applications} for
details.

\section{Preliminaries}\label{sec:prelims}

For any $k \in \mathbb N$, we write $[k]\eqdef \set{1,2,\ldots,k}$ to denote the set of the $k$ smallest positive
integers.

We focus on problems in which the exact solution is a single number.
Since we consider these problems in the streaming setting, it is usually not possible
to solve them exactly in small space and we aim to solve them approximately, staying within a multiplicative
factor close to 1 within the exact solution. This notion of approximation is captured
in the following definition.

\begin{definition}[Multiplicative approximation]
Let $\eps,x,y \in [0,\infty)$. We say that $y$ is a \emph{$(1+\eps)$--multiplicative approximation to
$x$} if $x \le y < (1+\eps)x$.
\end{definition}

Note that we used here for convenience a one--sided notion of approximation,
where we do not allow the approximate value $y$ to be smaller than $x$. However,
by standard shifting arguments, all results in this paper would hold, with small adjustments, if we
allowed for error in both directions, i.e., if we required that
$x / (1+\eps) < y < (1+\eps) x$ (and/or if we replace the strict inequality signs with $\leq$).

We now introduce the general form of problems in which we are interested.

\begin{definition}[Estimation problems]
\label{def:estimation_problems}
For a given function $f: \mathbb Z^n \to \set{0} \cup [1, \alpha]$ and parameter $\eps > 0$,
we refer to the task of computing a $(1+\eps)$--multiplicative approximation
to $f(v)$ on a given input $v \in \mathbb Z^n$ as the
\emph{$(f,\eps)$--estimation problem}. (To simplify our presentation, we require that $\alpha \ge 2$
so that $\log \alpha \ge \log 2 > 0$, i.e., $\log \alpha$ is not less than some positive constant.)
\end{definition}

We usually do not have an explicitly given vector $v$ that is an input
to an $(f,\eps)$--estimation problem and instead we work
with a stream of updates to the initial all--zero vector.

\begin{definition}[Frequency vectors and streams of updates]\label{def:streams}
We refer to inputs of our estimation problems as \emph{frequency vectors}.
For a given vector $v \in \mathbb Z^n$, $v_i$ is the \emph{frequency}
or \emph{number of occurrences} of item $i \in [n]$. (Note that we allow frequencies
to be negative.)

A frequency vector $v \in \mathbb Z^n$ can be defined via a \emph{stream of updates},
which is a sequence of $t \in \mathbb N$ pairs of the form
$(i_j,\Delta_j) \in [n] \times \set{-1,1}$ for $j \in [t]$.
The first element of the pair, $i_j$, describes the element whose frequency
is being updated. The second element describes whether the update
decreases (when $\Delta_j = -1$) or increases (when $\Delta_j = 1$) the frequency of $i_j$.
We refer to pairs $(i_j,\Delta_j)$ as updates.

A stream of updates of this form---or, in fact, any prefix of it---accumulates to a frequency
vector. For a stream of $t$ updates as above, for any $j \in \set{0}\cup[t]$,
we write $v^{(j)}$ to denote the frequency vector in $\mathbb Z^n$ obtained
by applying the first $j$ updates to the all--zero vector. Formally,
for all $i \in [n]$,
\begin{equation}\label{eq:frequency_vector}
    v^{(j)}_i \eqdef \sum_{\substack{1\le j' \le j\\i_{j'} = i}} \Delta_{j'}.
\end{equation}
\end{definition}

We now define (a version of) the standard streaming model in which the stream is fixed in advance.

\begin{definition}[``Oblivious'' streaming]\label{def:obliv_stream_alg}
A \emph{streaming algorithm in the oblivious adversary model} (or, for short, \emph{oblivious streaming algorithm})
for an $(f,\eps)$--estimation problem is given the parameter $\eps$ and an upper bound $m$ on the length
of the stream. Then, after receiving one by one $t \le m$ updates $(i_j, \Delta_j)$ for $j \in [t]$,
it produces a solution to the $(f,\eps)$--estimation problem on $v^{(t)}$,
which is the cumulative frequency vector defined in Definition~\ref{def:streams}.
\end{definition}

We now introduce the notion of \emph{order--invariant} oblivious streaming algorithms. Intuitively, the definition is relevant to problems in which the current value of $f$ on the underlying stream
depends only on the current frequency vector---$v^{(j)}$ at round $j$---and not on the \emph{order} of updates that accumulated to $v^{(j)}$. We note that similar concepts have been defined before---e.g., ``path independent automatons" \cite{li2014turnstile}---but we state the following version for completeness. Order--invariant algorithms are central in the streaming literature, since linear sketches are order--invariant.

\begin{definition}[Order invariant streaming algorithms \cite{li2014turnstile}]
\label{def:order_invariant}
Let $\alg$ be an oblivious streaming algorithm for an $(f,\eps)$--estimation problem
that processes streams
$S = ((i_1,\Delta_i),\ldots,(i_t,\Delta_t))$
of at most $t \le m$ updates.
Furthermore, suppose that $\alg$ uses an internal random seed $\rho \in \{0,1\}^{*}$
and that its memory after processing an entire stream $S$ of updates for a random seed $\rho$
is $\mathcal M(S,\rho)$. Let $v(S)$ denote the frequency vector that is the result
of accumulating all updates in $S$. (This corresponds to $v^{(t)}$ in Definition~\ref{def:streams}
for the given stream of updates of length $t$).

We say that $\alg$ is \emph{order invariant} if for any pair $S_1$ and $S_2$ of streams of length at most $m$
such that $v(S_1) = v(S_2)$ and any possible random seed $\rho$, we have $\mathcal M(S_1,\rho) = \mathcal M(S_2, \rho)$.
\end{definition}

We conclude by formally defining two parameters of streams, \emph{flip number} and \emph{density}, which have both proved important in the robust streaming literature. Indeed, most of the existing robust streaming frameworks have space complexity which depends on the flip number, and, in some cases, the density as well.

\begin{definition}[Flip number \cite{advrob}]
\label{def:flip_number}
Let $\eps \ge 0$, $m \in \mathbb{N}$, and $\bar{y} = (y_0, y_1, \ldots, y_m)$ be any sequence of real numbers. The $\eps$--flip number of $\bar{y}$ is the maximum $k \in \mathbb{N}$ such that there exists $0 \le i_1 < \ldots < i_k \le m$ so that $y_{i_j-1} \not \in
[(1-\eps)y_{i_j},(1+\eps)y_{i_j}]$
for every $j = 2, 3, \ldots, k$.
\end{definition}
Note that in our paper, the value of $\eps$ when defining the flip number will be clear from context.

\begin{definition}[Density/sparsity of a stream]
\label{def:sparsity}
For a vector $x \in \R^n$, we let $\|x\|_0 \eqdef |\{x_i \mid x_i \ne 0 \}|$. We refer to this quantity
as \emph{the density} (sometimes also referred to as the \emph{sparsity}) of vector $x$.

The \emph{density} (or \emph{sparsity}) of a stream of updates at time $j$ is the density of the frequency vector $v^{(j)}$, where $v^{(j)}$ is as defined in Definition~\ref{def:streams}.
\end{definition}

\section{The \BMAM}
\label{sec:model_definition}

Our adversarial streaming setting is formalized in the following definitions. It
builds on and is inspired by the definition of robust streaming,
which was introduced by Ben-Eliezer, Jayaram, Woodruff, and Yogev~\cite{advrob}.
The main difference here is that we carefully restrict the adversary's memory.

The model describes a game between two players, an adversary, which we often call $\Adv$, and a streaming algorithm, which we often call $\Alg$,
over $m$ rounds. For every round $j \in [m]$, the goal of $\Alg$
is to output an estimate $y_j$ to the value of a fixed function $f$
(i.e., it is a \emph{tracking} streaming algorithm, as opposed
to the algorithm in Definition~\ref{def:obliv_stream_alg},
which only outputs an estimate once, at the end of the stream).

Formally, we fix an integer $\alpha \ge 2$, let $\eps > 0$ be an approximation parameter, and let 
function $f: \mathbb{Z}^n \to \{0\} \cup [1, \alpha]$ be a function of the underlying frequency vector as defined in Equation~\eqref{eq:frequency_vector}. As a concrete example, $f$ could output the $F_p$--moment of the frequencies or the number of distinct elements. We assume both parties know $f$, $\alpha$, and $\eps$. The \emph{adversarial streaming model} is formalized as follows:

\begin{definition}[\BMAM]
\label{def:robust_streaming} The \BMAM{} is a game between two players, an
$\Adv = \Adv(r, k, \mathcal{P}, \mathcal{E}, \mathcal{W})$ and an $\Alg$.  $\mathcal{P}, \mathcal{E}$, and $\mathcal{W}$
are different types of memory accessible to $\Adv$, defined shortly. Initially, they are all blank. 
The game proceeds in $m$ rounds, corresponding to
values of $j \in [m]$ from $1$ through $m$, according
to the following rules.
For $j \ge 2$:
     \begin{itemize}
     \item $\Adv$ draws fresh randomness $r_j \in \{0,1\}^*$ with read--only access. After computing the update $(i_j, \Delta_j)$ (see next bullet), the memory containing $r_j$ is made blank.
    \item $\Adv$ computes an update $(i_j, \Delta_j) \in [n] \times \set{-1,1}$, which depends
    only on the randomness $r_j$, a \emph{persistent memory} $\mathcal{P}$ of at most $k$ bits, and a read--only \emph{estimate memory} $\mathcal{E}$ containing an estimate $y_{j-1} \in \{0, 1\}^*$. $\Adv$ has both read and write access to $\mathcal{P}$ (and, in particular, can write to $\mathcal{P}$ during its computation of $(i_j, \Delta_j)$),
    but only read access to $\mathcal{E}$. When computing $(i_j, \Delta_j)$, $\Adv$ also has read and write access to an unbounded \emph{working memory} $\mathcal{W}$. After computing update $(i_j, \Delta_j)$, the \emph{working memory} $\mathcal{W}$ is made blank.
    \item $\Alg$ receives $(i_j, \Delta_j)$, updates its internal memory state, and produces an estimate $y_j$.
    \item The \emph{estimate memory} $\mathcal{E}$ of $\Adv$ is then overwritten with $y_j$.
\end{itemize}
For $j = 1$, the update $(i_j, \Delta_j)$ is computed analogously, but it does not depend on any previous estimates (as there are none).

Let $v^{(j)}$, for $j \in [m]$, be the frequency vector accumulating the first $j$ updates as in Definition~\ref{def:streams}.
If in all rounds $j \in [m]$,
$y_j$ is a correct solution to the $(f,\eps)$--estimation problem
on $v^{(j)}$, then we say that $\Alg$ succeeds.
Otherwise, $\Alg$ fails.
Finally, we say that $\Adv$ is \emph{deterministic} if the random seeds are always set to the empty string.
\end{definition}

Note that the above definition only specifies the restrictions on $\Adv$ and does not explicitly restrict the space used by $\Alg$.
Our theorems provide restrictions on the space required by $\Alg$, which is the main computational resource optimized in streaming algorithms.

Recall that the estimate memory $\mathcal{E}$ models the fact that the adversary has access to the previous output of $\Alg$. This easily extends to the case that $\Adv$ has access to any fixed number $t$ of previous outputs by specifying that the estimate memory $\mathcal{E}$ contains the last $t$ estimates output by $\Alg$. All results proved in this paper extend in a straightforward manner to this setting.

\begin{remark}[Discussion of randomness: fresh vs.\ persistent]
\label{remark:randomness}
Note that our model lets the adversary access fresh randomness in each round. Since the adversary has limited amount
of memory, storing outcomes of coin tosses across several rounds may not be possible. A natural question
in this setting is whether access to a persistent infinite random string, drawn before the interaction with the algorithm, could
make the adversary more powerful and more difficult to defend against. After all, the random string could, for instance,
be used to simulate a random oracle, which could then provide a fully independent hash function with no extra storage needed.

It turns out that this would \emph{not} make designing adversarially robust
algorithms more difficult. Our goal is to provide robust algorithms that output
correct estimates throughout the entire stream with probability at least
$1-\delta$, where $\delta$ is a parameter selected by the user of the algorithm.
It is not difficult to show---as we do in
Lemma~\ref{lem:persistent_randomness}---that if there is an adversary who
succeeds at breaking the estimates with probability greater than $\delta$, using
persistent randomness, a successful configuration of persistent randomness can be
baked into its internal logic, which makes it an adversary in our model, with no
persistent randomness.

Conversely, one could ask whether fresh randomness is needed, because it may
seem that the initial persistent randomness could be used to simulate it. This
would be true if the adversary had enough memory to keep track of which bits of
persistent randomness have already been used. For instance, this could be
achieved by maintaining the pointer to the first unused bit of randomness. Our
results are meaningful only if the adversary has access to a very small amount
of persistent memory.
In this setting, after the persistent randomness is
fixed, the adversary generates the same update for a given estimate
in the same state of memory. This makes it a deterministic adversary,
which is easier to analyze, because we can force it to produce
a low--sparsity frequency vector by rounding the algorithm's estimates.
More details on how this is achieved are provided when we discuss
our results for deterministic adversaries.
\end{remark}

\section{Deterministic Adversaries}
\label{sec:deterministic_adv}
In this section we prove our robust streaming result against deterministic bounded--memory adversaries. 
First, the following lemma shows that the possible outputs of $f$ admit a small ``$\eps$--net'' with respect to multiplicative approximation.
We defer its straightforward proof to the appendix.

\begin{lemma}\label{lem:net}
    Consider the interval $[1, \alpha]$
    for an integer $\alpha$ that satisfies $\alpha \ge 2$. Let  $\eps \in (0,1)$. There exists a subset $\mathcal{N} \subseteq [1, \alpha]$ of size $O(\eps^{-1}\log\alpha)$ with the property that
    \[\forall x \in [1, \alpha], \quad \exists y \in \mathcal{N} \text{ satisfying } x \le y < (1+\eps)x. \]
\end{lemma}

Our main result provides efficient streaming algorithms against deterministic memoryless adversaries. It is assumed that we are able to compute $f$ and round its value to the net produced by Lemma \ref{lem:net} using $M(f, \eps)$ bits by just specifying the non-zero coordinates of the input vector.

\begin{theorem}
\label{thm:deterministic_adversary}
Consider the \BMAM{}
(see Definition \ref{def:robust_streaming}) with a deterministic adversary and $k = 0$, i.e.,
no persistent storage. There exists a deterministic streaming algorithm that
always succeeds and uses $O\left(\frac{\log(\alpha)}{\eps} \cdot \log(mn) + M(f, \eps)\right)$ bits of space.
\end{theorem}

\begin{proof}
We construct a deterministic algorithm, $\Alg$, and analyze its correctness and space complexity. Its description is as follows:
\begin{enumerate}
    \item The algorithm maintains a \emph{sparse} representation of the \emph{exact} frequency vector $v$, which equals $v^{(j)}$ at time $j$. More specifically, it explicitly keeps a mapping from
    all coordinates $i \in [n]$ that have received at least one update to their current frequencies. This is possible because the frequencies of all other coordinates equal 0.
    Initially, the frequency vector $v$ is the all zero vector---$v \leftarrow (0,0,\ldots,0)$---so the mapping is empty.

    \item At any round $j \in [m]$, upon receiving an update $(i_j, \Delta_j)$, it simply updates the frequency of coordinate $i_j$ in the mapping it maintains by setting
    it to $v_{i_j} \leftarrow v_{i_j} + \Delta_j$. If the mapping from $i_j$ is not present in the mapping, it is first added with the initial frequency of 0.

    \item To produce estimate $y_j$, $\Alg$ computes and outputs
    a rounded version of the value $f(v^{(j)}) \in  \{0\} \cup [1, \alpha]$.
    This process can be seen as first computing $f(v)$, where $v = v^{(j)}$
    is the current frequency vector, and then rounding it up to the smallest
    element that is at least $f(v)$ in the net $\mathcal N \cup \set{0}$
    given by Lemma~\ref{lem:net}, before outputting it.
    Thus, we output $y_j \in \mathcal N \cup \set{0}$ satisfying $f(v^{(j)}) \le y_j < (1+\eps)f(v^{(j)})$ and $y_j \in \mathcal{N}$.
    We know that this can be achieved,
    using at most $M(f,\eps)$ space.
\end{enumerate}

The correctness of $\Alg$ follows from the fact that it can always calculate the
\emph{exact} value of $f(v^{(j)})$ and round it up to a value $y_j$ in the net
$\mathcal{N} \cup \{0\}$ from Lemma \ref{lem:net}, which by definition
is a $(1+\epsilon)$--multiplicative approximation to $f(v^{(j)})$.
For the space complexity, note that $\Adv$ sees at most $|\mathcal{N}| + 1 = O(\eps^{-1}\log(\alpha))$ unique responses from $\Alg$. Since $\Adv$ is deterministic (and its persistent memory and random string are always blank), it follows that it is only a function of its estimate memory $\mathcal{E}$, which can only take on
at most $|\mathcal N| +1$ different values.
Thus, $|\{i_1, \ldots, i_m \} |$, where $i_j$ is the update of $\Adv$ at round $j$, is bounded by $|\mathcal N| + 1 + 1 = O(\eps^{-1}\log(\alpha))$,
where the additional additive term of $1$ comes from the fact that, initially, the adversary produces a
(potentially different) single update that does not depend on any estimate, because there is no estimate initially.
That is, the total number of unique items that the adversary places in the stream (either for insertions or deletions) is also at most $O(\eps^{-1}\log(\alpha))$. Thus, the total number of coordinates that $\Alg$ ever has to maintain is also bounded by the same quantity, leading to the claimed space bound. Note the $\log m$ term in the $\log (mn) = \log m + \log n$ overhad comes from the fact that a single frequency that $\Alg$ maintains could be any integer value between $-m$ and $m$. The $\log n$ term comes from the fact
that we also need to store the elements in $[n]$ that $\Alg$ has encountered. This completes the proof.
\end{proof}

The proof easily extends to the case of $k > 0$, when $\Adv$ has $k$ bits of persistent memory. In this case, $\Adv$'s outputs (which are the stream updates) are a deterministic function of at most
$O(\eps^{-1}\log\alpha) \cdot 2^k$ different possible values.

\begin{corollary}
\label{thm:deterministic_adversary2}
Consider the \BMAM{} (Definition \ref{def:robust_streaming}) with a deterministic adversary and persistent memory of $k$ bits. There exists a deterministic streaming algorithm that always succeeds and uses $O\left(2^k \cdot \frac{\log(\alpha)}{\eps} \cdot \log(mn) + M(f,\eps)\right)$ bits of space.
\end{corollary}

\section{Memoryless Adversaries Can Generate Non-Trivial Streams}

In this section we prove the following result, regarding the ability of memoryless or extremely low--memory adversaries to generate streams that are both non-sparse and have large flip number.

\begin{theorem}
\label{thm:lb_construction}

Fix $0 < c < 1/2$. Let $m,n \in \N$ and $\eps \in (0,1]$,\footnote{Note that we
allow $\eps$ to be arbitrarily small. In particular, it can depend on parameters
such as $m$ and $n$.} where $n > 10 m^{2c}$. Consider an instantiation of the
\BMAM{} for the $(F_2, \eps)$--estimation problem. There exists a randomized
memoryless $\Adv$ that, conditioned on $\Alg$ correctly solving the
$(F_2,\eps)$--estimation problem over the entire stream, with probability at
least $9/10$, generates a stream of length $m$ whose density is $\Omega(m^c)$
for all but the first $m^c$ rounds of the stream, and whose flip number is
$\Omega(m^{1-c})$.

In addition, if we provide $\Adv$ with one bit of persistent memory (i.e., $k=1$ in Definition \ref{def:robust_streaming}), then the above statement holds, for any constant $0 < c < 1$, with flip number $\Omega(m^{1-0.5 c})$.

\end{theorem}

\begin{proof}
We start with the proof of the second part of the theorem statement, in which $\Adv$ has \emph{one bit} of persistent memory.
The strategy and proof required for a memoryless adversary are slightly more complicated, and we discuss them later in the proof.

For convenience, we assume that $m^{c}$ is an integer; if not, we can replace it with $\lceil m^{c} \rceil$ as needed.

\paragraph{When $\Adv$ has one bit of persistent memory.}
Consider the following adversarial strategy requiring one bit of persistent memory. $\Adv$ uses the persistent memory bit to track in which of two possible
states it is. These states are
$\uparrow$ (``going up'') and $\downarrow$  (``going down''). The initial state is $\uparrow$.
\begin{itemize}[leftmargin=*]
\item \textbf{Type I insertion}: If the previous $F_2$ estimate sent by $\Alg$ is smaller than $m^{c}$ (or if it is the first round of the game and no previous estimate is available), $\Adv$ picks a uniformly random element $i$ from $\{2,3,\ldots,n\}$ and \emph{inserts} $i$. The state of $\Adv$ is set to $\uparrow$.
\item \textbf{Type II insertion}: If the previous $F_2$ estimate of $\Alg$ is at least $m^{c}$ and less than $(1+\eps) m^{c}$, $\Adv$ \emph{inserts} the element $1$ and sets its state to $\uparrow$.
\item \textbf{Deletion}: If the previous $F_2$ estimate of $\Alg$ is at least $(1+\eps)^3 m^{c}$, $\Adv$ \emph{deletes} the element $1$ and sets the state to $\downarrow$.
\item \textbf{Following the sign:} If the previous $F_2$ estimate is at least $(1+\eps) m^c$ and less than $(1+\eps)^3 m^c$, then $\Adv$ \emph{inserts} $1$ if its internal state is $\uparrow$, and \emph{deletes} $1$ if its internal state is $\downarrow$. The internal state remains unchanged after this operation.
\end{itemize}
Note that $\Adv$ can be implemented in the \BMAM{} with one bit of persistent memory. Indeed, the action of $\Adv$ at any given time depends only on the following three sources: (i) the previous estimate of the algorithm, (ii) whether the previous state is $\uparrow$ or $\downarrow$, and (iii) a random sample from
the uniform distribution on $\{2,\ldots,n\}$. It remains to prove that this adversarial strategy satisfies the statement of the theorem.
In the following, we write $\fv{t}$ to denote the frequency vector after $t$ updates, i.e., we use the notation from Definition~\ref{def:streams}.

We now prove a sequence of claims about the nature of the game given the above adversarial strategy. We stress that these claims hold \emph{for any} $\Alg$ (even
one that has unlimited memory), as long as $\Alg$ correctly solves the $(F_2, \eps)$--estimation problem after each update.

\begin{claim}
\label{claim:type_I_insertions}
There are at most $m^c$ Type I insertions throughout the stream.
\end{claim}

\begin{proof}
Once $m^{c}$ insertions of this type have been made, the actual value of $F_2$ remains at least $m^{c}$,
because elements from $\{2,3,\ldots,n\}$ are never deleted.
By the correctness of estimates,
the estimate visible to $\Adv$ must also be at least $m^{c}$. This implies that no
Type I insertion can occur in subsequent rounds.
\end{proof}

Consider the event that no element from $\{2,3,\ldots,n\}$ was inserted twice throughout the stream; let us call this event \emph{unique insertion}. Note that these elements are only added via Type I insertions, and are never deleted. As discussed above, at most $m^{c}$ type I insertions can occur throughout the stream. Via a standard birthday paradox argument, unique insertion happens with probability at least $19/20$ if $n \geq 10(m^{c})^2 = 10 m^{2c}$.

For the rest of the proof, we condition on the unique insertion event (which, as mentioned, occurs with probability at least $19/20$). Assuming unique insertion, the sparsity part of the theorem easily follows.

\begin{claim}
\label{claim:high_density}
After the first $m^c$ rounds of the adversarial game (and throughout the rest of the stream), the density of the stream is always at least $m^c / (1+\eps) = \Omega(m^c)$.
\end{claim}
\begin{proof}
As long as the density of the stream is at most $m^c / (1+\eps)$, $\Alg$ must return an output smaller than $m^c$. Thus, in the first $1 + \lfloor m^c / (1+\eps) \rfloor$ rounds of the stream, all actions by $\Adv$ are Type I insertions. These inserted elements are all distinct, assuming unique insertion, and are never deleted. Thus the density remains at least $m^c / (1+\eps)$ for the rest of the game.
\end{proof}

The rest of the proof establishes the flip number lower bound (which occurs with probability $1$, conditioned on unique insertion). We start with the following claim.

\begin{claim}
\label{claim:frequency_nonnegative}
The frequency of $1$ in the stream, at any point throughout the game, is non-negative.
\end{claim}
\begin{proof}
Suppose that this is not the case.
Then there exists some minimal time $t$ after which the frequency of $1$ becomes negative for the first time. Thus, at round $t+1$, $\Adv$ chooses to delete 1, meaning that we are in either the ``Deletion'' or ``Following the sign'' regime. This, in turn, implies that $\Alg$'s estimate at time $t$ is at least $(1+\eps)m^{c}$.

On the other hand, the frequency of $1$ at time $t$ must be zero (since it turns negative after the deletion). By Claim \ref{claim:type_I_insertions} and unique insertion,
$\ft{t} \le m^c$.
By correctness, $\Alg$'s estimate at time $t$ must be less than $(1+\eps)m^{c}$. This stands in contradiction to the first paragraph.
\end{proof}

Let $S$ denote the set of all times throughout the stream at which a Type I insertion is made. By Claim~\ref{claim:type_I_insertions}, $|S| \leq m^{c}$. Our last two claims of the proof are given next.

\begin{claim}
\label{claim:going_up}
Let $t \leq m - (1+\eps)^4m^{c}$ and suppose that (i) $\ft{t} \leq (1+\eps)m^c$
and (ii) the internal state of $\Adv$ at time $t$ is $\uparrow.$

Then there exists $t' \in [t+1,m]$ such that
(i) $\ft{t'} > (1+\eps)^2 m^{c}$, (ii) the internal state of $\Adv$
is $\downarrow$ at time $t'$ and $\uparrow$ at time $t'-1$,
and (iii) the number of values $t'' \in [t,t')$ \emph{not belonging} to $S$ is at most $O(m^{c/2})$.
\end{claim}
\begin{claim}
\label{claim:going_down}
Let $t \leq m-(1+\eps)^2 m^{c/2}$ and suppose that (i) $F_2(t) > (1+\eps)^2m^c$ and (ii) the internal state of $\Adv$ at time $t$ is $\downarrow$, and at time $t-1$ it is $\uparrow$. There exists a time $t'$
such that $t < t' < t + O(m^{c/2})$, $\ft{t'} \leq (1+\eps) m^c$, and the internal state at time $t'$ is $\uparrow$.
\end{claim}

With the above two claims in hand, one can see that the flip number of the stream is $\Omega(m^{1-c/2})$. Indeed, when disregarding all rounds $t \in S$ of the game (note that there are only $O(m^c)$ such rounds), the above two claims show that the output of the algorithm must flip---more specifically, either increase or decrease by a factor of $1+\eps$---every $O(m^{c/2})$ rounds, and so the flip number is at least
$$ \frac{m-O(m^c)}{O(m^{c/2})} = \Omega\left( m^{1-\frac{c}{2}} \right).$$
We now prove these claims. 

\begin{proof}[Proof of Claim~\ref{claim:going_up}]
Consider any time $t$ satisfying the conditions of the claim. Let $t' > t$ be the minimal integer for which (i) $\ft{t'} > (1+\eps)^2 m^{c}$, (ii) the internal state of $\Adv$ at time $t'$ is $\downarrow$, and at time $t'-1$ it is $\uparrow$. We note that such $t'$ must exist as long as $t \leq m - (1+\eps)^4 m^{c}$, since all updates between $t$ and $t'$ are insertions and all entries of the frequency vector at time $t$ are non-negative.

For any $t'' \not\in S$ such that $t \le t'' < t'$,
the action taken by the adversary in round $t''+1$ is either a Type II insertion or a ``Follow the Sign'' step with an $\uparrow$ state (both of which lead to inserting $1$). Indeed, if $\ft{t''} \leq (1+\eps)^2 m^c$, then the correctness of $\Alg$ implies that the estimate at time $t''$ is smaller than $(1+\eps)^3 m^c$; and otherwise, the minimality of $t'$ implies that the sign at time $t''$ must be $\uparrow$. In both cases, the internal state remains $\uparrow$, as the only action that changes it to $\downarrow$ is a deletion.

Finally, $\Theta(m^{c/2})$ insertions of $1$ suffice to increase the $F_2$ value of the stream by $\Omega(m^{c})$, thus reaching $t'$ as above within $O(m^{c/2})$ rounds (disregarding rounds with Type I insertions).
\end{proof}

\begin{proof}[Proof of Claim~\ref{claim:going_down}]
From the condition that the sign at time $t-1$ is $\uparrow$, we know that $\ft{t-1} \leq (1+\eps)^3 m^c$. In particular, the contribution of the element $1$ to $\ft{t}$ is $O(m^c)$, meaning that the frequency of $1$ at time $t$ is $O(m^{c/2})$. Let $t' > t$ denote the first time at which an insertion of Type I or Type II occurs. Note that such a time $t'$ must exist, provided $t < m - O(m^{c/2})$, and in addition, $t' < t+O(m^{c/2})$. Indeed, as long as we only have deletions or follow the sign rounds, the internal state remains $\downarrow$, leading to deletions of the element $1$. Now, because the frequency of $1$ is $O(m^{c/2})$ at time $t$, after $O(m^{c/2})$ rounds we reach a situation in  which either a type I or II insertion happens, or the frequency of $1$ becomes zero, in which case the $F_2$ value is at most $m^{c}$ and the next step must either be a type I or II insertion (depending on the estimate output by $\Alg$).

Now, since $\Adv$ made a type I or type II insertion at time $t'$, the estimate of $\Alg$ at time $t'$ must be at most $(1+\eps) m^c$. Thus, the actual $F_2$ value is also at most $(1+\eps) m^c$. The proof follows.
\end{proof}

\paragraph{When $\Adv$ is memoryless.}
The strategy of $\Adv$, as well as the analysis, are quite similar to the above case. We detail below the modifications required in the strategy and proof.

\begin{itemize}
\item \textbf{Type I insertions}, \textbf{Type II insertions}, and \textbf{deletions} remain the same, except that that the deletion regime only applies if the previous estimate is at least $(1+\eps)^4 m^c$ (not $(1+\eps)^3 m^c$ as previously defined). In addition,
we now no longer have the internal state (of $\uparrow$ or $\downarrow$). Note that in all of these three action types, the action chosen by the adversary with one bit of permanent memory did not depend on the internal state, and so these remain well-defined.
\item The ``follow the sign'' action is replaced with an \textbf{(unbiased) random walk} action: If the previous $F_2$--estimate is at least $(1+\eps)m^c$ and less than $(1+\eps)^4 m^c$, then $\Adv$ \emph{inserts 1 with probability 1/2}, and \emph{deletes 1 with probability 1/2}.
\end{itemize}

This strategy by $\Adv$ is clearly memoryless: each step depends only on the last estimate provided by $\Alg$, where in particular $\Adv$ may sample the action from a fixed distribution (depending on the last estimate). Claims \ref{claim:type_I_insertions}, \ref{claim:high_density}, and \ref{claim:frequency_nonnegative} still hold (and their proofs remain unchanged), again conditioning on the unique insertion event, which holds with probability at least 19/20. It remains to prove the flip number lower bound.

We next give a result analogous to Claims \ref{claim:going_up} and \ref{claim:going_down} in the memoryless case. Its proof relies on standard hitting time results on random walks in one dimension, together with a simple coupling argument.
\begin{claim}
\label{claim:random_walk_lower_bound}
There exists an absolute constant $\alpha > 0$ satisfying the following. For every $t$ such that $0 \leq t \leq m-\alpha \cdot m^{c}$, there exists a $t'$ such that $t < t' \leq t+\alpha \cdot m^{c}$, $\ft{t'} \notin (1 \pm \eps)\ft{t}$, with probability at least $1/2$, independently of all randomness revealed before round $t$ of the game.
\end{claim}
We now prove Claim~\ref{claim:random_walk_lower_bound}. For the proof, we use a standard hitting time result on random walks in one dimension. An unbiased random walk in 1D refers to the following process: start at location $0$ on the one--dimensional grid. In each round, move one step to the right (i.e., from location $s$ to $s+1$, where $s$ is the current location) with probability $1/2$ (independently of all previous steps), and one step to the left (i.e., from $s$ to $s-1$) otherwise. The hitting time of the random walk with respect to a value $s$ is the first round in which the random walk arrives at location $s$.
\begin{lemma}[Hitting time of 1D random walk; see, e.g., \cite{aldous-fill-2014}, Chapter 5]
\label{lem:1D_rand_walk}
The hitting time for any location $s \in \mathbb{Z} \setminus \{0\}$ in an unbiased 1D random walk is with probability $1/2$ at most $\beta s^2$, where $\beta > 0$ is an absolute constant.
\end{lemma}

\begin{proof}[Proof of Claim~\ref{claim:random_walk_lower_bound}]
There are two cases that one should consider. The first case is $\ft{t} < (1+\eps)^2 m^c$, and in this case we prove the statement of the claim for some $t'$ satisfying that $\ft{t'} \geq (1+\eps)^3 m^c > (1+\eps) \ft{t}$. The second is the complementary case, in which  $\ft{t} \geq (1+\eps)^2 m^c$, and we prove the analogous statement for some $t' < (1+\eps)m^c \leq \ft{t} / (1+\eps)$. These cases are very similar, as we shall explain soon. For now, let us consider the first case.

Consider a coupling of the adversarial game to a random walk, where each insertion of $1$ is viewed as a step up, and each deletion of 1 is a step down (and we think of insertions of elements other than $1$ as not changing the state). By Lemma \ref{lem:1D_rand_walk}, with probability at least $1/2$ an unbiased random walk will hit location $(1+\eps)^{3/2} m^{c/2}$ in time $O((m^{c/2})^2) = O(m^{c})$. Now what about the adversarial game, starting at time $t$, and as long as we do not exceed the $F_2$ value of $T \eqdef (1+\eps)^3 m^c$? It turns out the the frequency of $1$ in this case is not an unbiased random walk, but it actually dominates the unbiased case: The probability to step up can be $1$ in some cases, while the probability to step down is never higher than $1/2$. Indeed, in this regime there are three possible actions: type I and type II insertions, and random walk operations.

There are at most $m^c$ rounds in total in which type I insertions occur. Type II insertions always lead to an insertion, i.e., a step up. Random walk steps lead to an insertion of $1$ with probability $1/2$ and deletion with probability $1/2$. Thus, we see that in all but $O(m^c)$ rounds (in which type I insertions occur, which do not change the frequency of the element 1), the adversarial game dominates an unbiased random walk, as long that the $F_2$ value has not exceeded $T$. This implies that with probability $1/2$, one of the following must hold between time $t$ and $t+O(m^c)$: either we reach an $F_2$ value of at least $T = (1+\eps)^3 m^c$, leading to a flip (with respect to time $t$), or the frequency of $1$ in the adversarial game reaches $(1+\eps)^{3/2} m^{c/2}$, making the $F_2$ value reach $T$ and causing a flip again. This concludes the first case.

The analysis of the second case is completely symmetric, except that we set the threshold at $T \eqdef (1+\eps) m^c$, and there are only two possible operations: either a deletion (leading to a step down with probability 1) or a random walk action. In both cases, the locations is dominated by an unbiased random walk, and again by an application of Lemma \ref{lem:1D_rand_walk}, the statement of the claim follows.
\end{proof}

With the claim in hand, it is straightforward to show an $\Omega(m^{1-c})$ flip number lower bound. Indeed, consider a sequence of all values of $t$ which are multiples of $\lfloor \alpha \cdot m^{c}\rfloor$ between $1$ and $m$; there are $\Omega(m^{1-c})$ such values. The probability of a flip between each consecutive pair of $t$ values in this sequence is at least $1/2$, independently of all previous information, and so the expected number of flips is $\Omega(m^{c})$. It follows, e.g., from Chernoff bound, that the total number of flips will be at least half its expectation with probability $1-2^{m^{\Omega(1)}} > 0.95$, where the inequality holds for large enough $m$. This completes the proof.
\end{proof}

\section{Randomized Adversaries}

\subsection{The \tSAM}

As discussed before, in the oblivious adversary model, the adversary selects the entire stream in
advance---without any interaction with the streaming algorithm---and is only required to provide a correct estimate at the end of the stream. There is a standard way to turn a correct algorithm in the oblivious setting into an algorithm with a \emph{tracking} guarantee (i.e., an algorithm that provides a correct estimate at any point along the stream) as long as the adversary remains oblivious.
The standard approach is first to amplify the
probability of the algorithm's success at any point, by using the median of
estimates from a small number of copies of the algorithm.
Then one applies the union bound to claim that the algorithm is unlikely to err at any point.

We now consider a generalization of the oblivious adversary model,
which we call the \tSAM{}. This model introduces adaptivity by allowing the
adversary to select $\tau$ streams in advance (obliviously) and arbitrarily
(and adaptively) merge them into the actual stream passed to the algorithm. We
describe it in relation to the \BMAM{}, introduced in
Definition~\ref{def:robust_streaming}. This time, $\Adv$ is allowed to use an
unlimited amount of memory, but it is a great question if restricting it could
lead to better robust algorithms in this model (see Section
\ref{sec:open_problems}).

\begin{definition}[\tSAM]\label{def:k_stream_adversary}
Let $\tau$ be a positive integer.
The \emph{\tSAM} is a game between
two players, $\Alg$ and $\Adv$, with the same interaction scheme
as in the \BMAM{}: in each round, $\Adv$ sends an update
and $\Alg$ replies with an estimate.
As opposed to the \BMAM{},
$\Adv$ is allowed unlimited persistent memory, but
is restricted in its choice of the stream it generates as follows:
\begin{enumerate}
 \item Initially, before the first round, $\Adv$ generates the first update, to be used in round 1. In addition, $\Adv$ generates
 $\tau$ streams of $m-1$ updates each at this point.
 \item In rounds 2 through $m$, $\Adv$ can use the complete knowledge it has collected thus far, including the latest estimate
 provided by $\Alg$, to
 select one of the $\tau$ streams arbitrarily. Then $\Adv$ removes the first element of that stream and passes
 it to $\Alg$ as the next update.
\end{enumerate}
\end{definition}

Note that thanks to the unlimited persistent memory, $\Adv$ can record the entire history of stream updates and algorithm outputs. The only restriction on $\Adv$
is the fact that it has to draw the next update at any time from one of the $\tau$ streams it fixes in advance.

\subsection{Robust Algorithms for the \tSAM}

We now describe an efficient algorithm for this
model, in which we apply the standard median trick
to increase the probability of success of a streaming algorithm
in the oblivious setting.
This is captured in the following lemma, whose proof we include in the appendix for completeness.

\begin{lemma}[Standard amplification trick]\label{lem:amplification}
Let $\alg$ be an oblivious streaming algorithm that, with probability
at least $9/10$, produces a correct solution to a $(f,\eps)$--estimation problem.
One can obtain an oblivious streaming algorithm that provides a correct
solution to this problem with probability $1-\delta$, for any $\delta \in (0,1/10)$,
by running $t \eqdef \lceil 12\ln(1/\delta)\rceil$
independent copies of $\alg$ and outputting the median of their estimates.
\end{lemma}

We use the above lemma to prove that there is a robust algorithm for the $\tau$--Stream Adversary
Model, whose space complexity scales linearly with $\tau$ (so it is efficient when $\tau$ is small). Here we need the additional condition that $\alg$ is \emph{order--invariant} (see Definition \ref{def:order_invariant}).
More specifically, we show
that an order--invariant oblivious streaming algorithm
is very likely
to provide correct estimates throughout the entire stream
if its probability of success has been amplified
so that it errs with probability much smaller than
$m^{-\tau}$.

\begin{lemma}\label{lem:k_stream}
Let $\alg$ be an order--independent oblivious streaming algorithm
that uses $M(\epsilon, \delta)$ space to provide a correct solution
to an $(f,\epsilon)$--estimation problem with probability $1-\delta$
for any $\epsilon \in (0,1)$ and $\delta \in (0,1)$.
There is a robust streaming algorithm
for the same problem in the \tSAM{} that uses
$M(\epsilon,\delta/m^\tau)$ space,
where $\epsilon \in (0,1)$ and $\delta \in (0,1)$,
and succeeds with probability at least $1-\delta$.
The space of the robust algorithm can also be bounded by
$O(\tau \log m + \log(1/\delta) ) \cdot M(\epsilon,1/10)$.
\end{lemma}

\begin{proof}
The only streams on which $\Alg$ may be asked to provide an estimate
are a result
of merging prefixes of the initial $\tau$ streams and appending to the first update.
Both the initial $\tau$ streams and the first update are selected without any
interaction with $\Alg$, i.e., non-adaptively.
Note that there are at most $t \eqdef m^\tau$ ways in which prefixes of the $\tau$ streams
can be selected to amount to streams of length between 1 and $m$,
which are the lengths on which $\Alg$ has to provide estimates.
While the selected prefixes can potentially be merged
in many different ways into a stream that $\Alg$ processes,
$\alg$ is order--independent and it therefore
behaves in an exactly the same way for a given selection of prefixes, independently
of the order in which it processes the updates they contain.
Hence, to leverage $\alg$, it suffices to ensure that it provides a correct estimate on one arbitrary
ordering of elements in the selected prefixes to provide
a correct ordering for all other orderings.
Furthermore, it suffices that $\alg$ produces a correct estimate
for each selection of prefixes of the initial $\tau$ streams
to ensure that it cannot be forced to provide an incorrect estimate
by $\Adv$.
This can be achieved by making the probability of $\alg$'s
failure on any specific selection of prefixes at most $\delta / t$. Then, by the union
bound, the probability that $\alg$ errs on any of them is at most
$t \cdot \delta/t = \delta$. Otherwise, all estimates throughout the stream
are correct, independently of what $\Adv$ does.
Hence, a robust $\Alg$ in the \tSAM{} can be achieved by
using $\alg$ and setting its error probability to at most
$\delta / m^\tau $ on any fixed input.
This requires $M(\epsilon, \delta/m^\tau)$ space.

We can also express the space requirement as a function of $M(\eps,1/10)$
by invoking Lemma~\ref{lem:amplification}.
We run $\lceil 12 \ln(m^\tau/\delta) \rceil$ independent
copies of $\alg$, each erring with probability at most $1/10$
and using $M(\eps,1/10)$ space. The median of their estimates
is a correct solution to the $(f,\eps)$--estimation problem
in the oblivious setting
with probability at least $\delta/m^\tau$.
Due to our analysis, this a robust algorithm in \tSAM{}
and uses $O(\tau \log m + \log(1/\delta)) \cdot M(\eps,1/10)$ space.
\end{proof}

\subsection{Robust Algorithm for Bounded Memory Adversaries}
We now apply Lemma \ref{lem:k_stream}
to construct a robust algorithm in the \BMAM{}.

\begin{theorem}\label{thm:robust_low_memory}
Let $\alg$ be an order--independent oblivious streaming algorithm
that uses $M(\epsilon, \delta)$ space to provide a correct solution
to an $(f,\epsilon)$--estimation problem with probability $1-\delta$
on streams of at most $m$ updates
for any $\epsilon \in (0,1)$ and $\delta \in (0,1)$.
For any $\epsilon \in (0,1)$ and any $\delta \in (0,1/10)$,
there is a robust streaming algorithm for the $(f,\epsilon)$--estimation problem
in the \BMAM{} with $k\ge 0$ bits of persistent memory,
that succeeds with probability at least $1-\delta$ and uses
$M(\epsilon/3,1/10) \cdot O\left(2^k \cdot \frac{\log \alpha}{\eps}\cdot \log m + \log(1/\delta)\right)$ space, where
$f$'s range is in $\set{0} \cup [1, \alpha]$ for $\alpha \ge 2$.
\end{theorem}

\begin{proof}
First consider $\alg$ with its parameters adjusted so that it solves the
$(f,\eps/3)$--estimation problem with probability at least $9/10$ on a stream of
length at most $m$. When it succeeds, it produces an estimate $y$ such that
$y_\star \le y < (1+\eps/3)y_\star$, where $y_\star$ is the exact value of $f$
on the final frequency vector. We modify $\alg$ to round $y$ to
a small set of options before outputting it. In particular, Lemma~\ref{lem:net} with
its parameter $\epsilon$ set to $\epsilon/3$ gives us a small set $\mathcal N$
of values such that every correct estimate for the $(f,\eps/3)$--estimation problem
can be rounded up to a value $y'$ in $\mathcal N_\star \eqdef \set{0} \cup \mathcal N$.
The size of $\mathcal N$---and therefore, also of $\mathcal N_\star$---is of order
$O(\eps^{-1}\log \alpha)$.
More specifically, the rounding works as follows:
\begin{itemize}
 \item If $y \in (-\infty,1)$, we set $y' = 0$.
 \item If $y \in [1,\alpha]$, we round it up to the closest value in $\mathcal N$ that is greater or equal to it.
 Such a~value exists due to the properties of $\mathcal N$, which was given by Lemma~\ref{lem:net}.
 This is the value to which we set $y'$ in this case.
 Due to the properties of $\mathcal N$, we have
 $y \le y' < (1+\eps/3) y$.
 \item Finally, if $y > \alpha$, we treat it as $\alpha$ and handle as in the previous case (which actually means outputting $y' = \alpha$).
\end{itemize}
We now claim that if $y$ is a correct estimate, by which we mean that it is a $(1+\eps/3)$--multiplicative approximation to $y_\star$,
then $y'$ is also a correct estimate, by which we mean that it is a $(1+\eps)$--multiplicative approximation to $y_\star$.
To see this consider different scenarios.
First, if $y \in (-\infty,0) \cup (0,1)$, $y$ is not a correct
estimate, so we can output anything, and we output $0 \in \mathcal N_\star$.
Second if $y =0$ and it is the correct estimate, $y_\star$ has to be 0 as well,
and since we output $y' = 0$, we output the exact value of $f$.
Next if $y \in [1,\alpha]$ and it is a correct estimate,
we round it up to a value $y' \in \mathcal N$ such that $y \le y' \le (1+\epsilon/3)y$.
In this case, on the one hand, $y_\star \le y \le y'$,
and on the other, $y' < (1+\eps/3) y < (1+\eps/3) (1+\eps/3)y_\star < (1+\eps)y_\star$,
which means that $y'$ that the algorithm outputs is
a $(1+\epsilon)$--multiplicative approximation. Finally, if $y > \alpha$ and it is a correct estimate,
then $\alpha$ is also a correct estimate, i.e., a $(1+\eps/3)$--multiplicative approximation to $y_\star$,
because the range of possible values of $f$ is bounded by $\alpha$ from above.
In this case, we can use $y = \alpha$ instead and this case reduces to the previous one.
In summary, $y'$ is a $(1+\eps)$--multiplicative approximation to $y_\star$
with probability at least $9/10$ and it always belongs to set $\mathcal N_\star$.

It follows that the modified version of $\alg$ can produce
at most $|\mathcal N_\star| = O(\eps^{-1} \log \alpha)$ different estimates that can be passed to $\Adv$.
Additionally,
the persistent memory of $\Adv$ can be in at most $2^k$ different states, which
means there are at most $O(2^k \eps^{-1} \log\alpha)$ different
input states for $\Adv$ before randomness is considered. Hence,
for a given input state, $\Adv$ always generates an update from
the same distribution. We can then reduce this scenario to
the $\tau$--Stream Adversary Model for $\tau = O(2^k \eps^{-1} \log\alpha)$
by observing that we could pre-generate a stream of $\Adv$'s updates
in each possible input state. Then any possible sequence of $\Adv$'s updates
is a result of merging these streams after the first initial update.
Hence, Lemma~\ref{lem:k_stream} yields the desired result.
\end{proof}

\subsection{Applications}\label{sec:applications}

\paragraph{Moment estimation.}
The most important application of our framework is to moment estimation.
Note that for any fixed power $p \geq 0$, since there are at most $m$ updates,
the maximum value of the moment function is bounded by $m^p$, if $p\ge 1$, or
$m$, if $p \in [0,1]$. Best known turnstile streaming algorithms for moment estimation are based on
linear sketching, which implies that they are order invariant
and can be used in our framework.
More specifically, two papers by Kane, Nelson, and Woodruff~\cite{optimal_l0,optimal_small_lp}
show that, for any fixed constant $p \in [0,2]$, the $(F_p,\eps)$--estimation problem can be solved
via a linear sketching algorithm in
$\tilde O (\eps^2 \log m)$ bits\footnote{
The $\tilde O$ notation used here hides polylogarithmic factors,
i.e., $\tilde O(f(\eps,m))$ denotes $O(f(\eps,m) \log^c f(\eps,m))$
for some constant $c > 0$.
This upper bound has been simplified
using our assumption that all updates are unit updates and $n = \poly(m)$.
}
of space with probability $9/10$ in the oblivious streaming
model.
This leads to the following corollary of Theorem~\ref{thm:robust_low_memory}.

\begin{corollary}For any fixed $p \in [0,2]$,
the $(F_p,\eps)$--estimation problem on streams with unit turnstile updates can be solved
in the \BMAM{} against an adversary with access to $k$ bits of persistent memory by a robust algorithm that uses
$\tilde O(2^k \eps^{-3}\log^3 m)$ bits of space with probability $9/10$.
\end{corollary}

\noindent
(Note that this result easily extends to integer updates from a range
$\{-M,-M+1,-M+2,\ldots,M-1,M\}$, for any $M$. Updates from a range wider than $\set{-1,1}$,
which corresponds to unit updates, simply increase the range of possible values, i.e., $\alpha$.
More specifically, without going into all details, this comes at the cost
of replacing at most two of the $\log m$ factors with $\log (Mn)$.)

\paragraph{Triangle counting.}
Another application concerns graph streaming. We focus on the fundamental triangle counting problem: here we are given a turnstile stream of edge updates to an unweighted graph on $n$ nodes and must return an estimate to the number of triangles $T$. The stream can be viewed as updating a frequency vector of size $O(n^2)$ of edges. Trivially, $T = O(n^3)$. 

We use the classic algorithm of Bar-Yossef, Kumar, and Sivakumar~\cite{bar2002reductions} that solves this problem via a black--box reduction to $F_0$, $F_1$, and $F_2$ moment estimation of the aforementioned edge frequency vector. While it is not the state of the art classical streaming algorithm (see the work of Jayaram and Kallaugher~\cite{jayaram_et_al} for a thorough discussion of the problem), it highlights the usefulness of Theorem \ref{thm:robust_low_memory}.

\begin{corollary}We write $n$ to denote the number of vertices in the input graph.
    Consider the $({\triangles},\eps)$--estimation problem in which $\triangles(G)$ is the number of triangles in the input graph $G$, represented here as a frequency vector of edges of length $\binom{n}{2}$.

    This triangle estimation problem on turnstile streams of length at most $\poly(n)$ can be solved in the \BMAM{} against an adversary who has $k$ bits of persistent space by a robust algorithm that uses $\tilde{O}(2^k \eps^{-3}(n^3/T)^2)$ space and succeeds with probability $9/10$, where $T$ is the number of triangles.
\end{corollary}

\section{Discussion and Open Problems}
\label{sec:open_problems}

We conclude the paper with a few natural open problems stemming from our work. The most important question in our opinion is
whether one can avoid the exponential dependence on $k$, the size of the persistent memory.
We restate the main question of the introduction:

\begin{question}\label{question:main}
    Can we avoid the $2^{\Omega(k)}$ dependence in the size $k$ of the persistent memory in our results for deterministic adversaries in Corollary \ref{thm:deterministic_adversary2} and for order--independent algorithms in Theorem \ref{thm:robust_low_memory}? In particular, can we obtain dependence polynomial in $k$?
\end{question}

 At a high level, we incur this dependence because we need to handle all $2^k$ many possible different states of the persistent memory. Progress towards our main open question could also have downstream implications in the original robust streaming setting of Ben-Eliezer et al.~\cite{advrob} (but now with \emph{deletions}!), since a large persistent memory allows the adversary to remember more of the interaction with the streaming algorithm and correlate future inputs with many past estimates, thus approaching the power of the adversary of Ben-Eliezer et al.~\cite{advrob}.

 A related question concerns the $\tau$--Stream Adversary Model (Definition \ref{def:k_stream_adversary}), in which the adversary has unlimited persistent memory, but must choose
each input from one of $\tau$ streams generated in advance.
 We ask if limiting the adversary's persistent memory in this case, analogous to our main model, the \BMAM, can result in more
 space--efficient streaming algorithms.

 \begin{question}
     Can we improve the dependence on $\tau$ in Lemma \ref{lem:k_stream} if we assume that the adversary has bounded persistent memory? Is dependence sublinear in $\tau$ possible?
 \end{question}

Note that progress towards the above question would directly improve Theorem~\ref{thm:robust_low_memory} and Question~\ref{question:main}, since the proof of Theorem~\ref{thm:robust_low_memory} proceeds by reducing the randomized adversary (in the \BMAM{}) to one which interleaves between exponentially many (in the size of the persistent memory) different streams.

We now ask about improving the dependence on $\alpha$, the range of $f$, in our theorem statements. This is especially important for settings such as heavy--hitter detection, in which valid outputs include
many different subsets of the domain $[n]$.
Note that in the classic streaming setting, it is possible to obtain space bounds depending on the number of heavy hitters \cite{charikar2002finding}, so we ask an analogous question for our \BMAM.
\begin{question}\label{q:large_range}
    Can our results be extended to the case where the number of different valid outputs is large (such as the range of $f$ being exponential in $n$)?
    In particular, what can we say about the problem of identifying heavy hitters in the \BMAM?
\end{question}

In this case, we can think of $\alpha$ as being exponential in $n$. Our current
bounds do not apply here in a satisfactory way. For example, we can set $\alpha
= 2^n$ in our theorem statements, so our space bound reduces to linear in $n$,
which can trivially be satisfied by explicitly keeping track of all coordinates.
At the heart of it, the challenge comes from the fact that we cannot efficiently
perform ``output rounding" for heavy hitters as we did in, for instance, Theorem
\ref{thm:deterministic_adversary}. In particular, it is easy to construct an
adversary, even with very limited persistent memory, that can cycle through many
different subsets of answers. For example, suppose $\Adv$ gives $(1,1)$ as the
first update. Then $\Alg$ must output $\{1\}$. Then $\Adv$ gives $(2,1)$, so the
algorithm's output must be $\{1,2\}$. Then $\Adv$ gives $(1, -1)$, so the right
set of heavy hitters is $\{2\}$, and so on. In other words, by looking at the
very last output, $\Adv$ can determine which elements have a frequency of one
and then correspondingly add or delete elements to update the set of heavy
hitters.

The next question concerns our lower bound construction of Theorem \ref{thm:lb_construction}. The result and proof apply specifically to the case of $f=F_2$, i.e., the goal is to compute
the second moment of the stream.
It remains open whether a similar attack could be deployed against other important functions, such as $F_0$ (distinct elements).

\begin{question}
    Is it possible for an adversary in the \BMAM{} to enforce a stream with both high density and high flip number for $F_0$? In particular, is it possible for an adversary with a persistent memory of size $o(\log m)$ to make these quantities $m^{\Omega(1)}$?
    Conversely, what if we give the adversary $k = O(\log m)$ bits of persistent memory?
\end{question}
We suspect the answer is negative, since one item can only influence the number
of distinct elements by $1$. Thus intuitively, it seems that if the adversary has to
track frequencies of polynomially many elements, either exactly or
approximately. If the adversary is exactly keeping track of such polynomially
many elements exactly, then it requires polynomial space. However, the adversary
itself could be simulating a distinct elements algorithm, using its persistent
memory, in which case $O(\log m)$ space may suffice. We leave this as an
intriguing open question.

Next we turn to Theorem \ref{thm:deterministic_adversary}. In its proof, we explicitly kept track of the frequency vector, which was forced to be sparse due to the rounding performed by the streaming algorithm. One could imagine making the algorithm even more space efficient by approximately tracking the sparse frequency vector,  but  this will likely depend on the underlying properties of $f$ which we abstracted away. 

\begin{question}
    Is the space bound of Theorem \ref{thm:deterministic_adversary} the right bound? Can we classify a broad family of functions $f$ that use even smaller space, e.g., by approximately tracking the frequency vector?
\end{question}

Lastly, we ask if the techniques of Woodruff and Zhou~\cite{WZ2024} can deal with bounded memory adversaries for the specific problem of $F_2$ estimation.

\begin{question}
    What are the limits of the techniques of Woodruff and Zhou~\cite{WZ2024} for estimating $F_2$ in the \BMAM?
\end{question}

Our construction of a stream with high flip number and high density heavily relies on the use of a single heavy hitter.
Toward answering the above question, it would be interesting to understand whether this is inherent or just an artifact of the specific construction. If it is indeed an inherent feature of the memoryless adversarial setting, than this could mean that dense--sparse strategies such as those of Woodruff and Zhou~\cite{WZ2024}, which leverage deterministic
algorithms for heavy--hitter detection, could also be useful for achieving efficient algorithms that are robust against memoryless adversaries.

We wish to conclude the discussion by pointing out the problem of \emph{missing item finding} (MIF), which has been explored by Stoeckl and Chakrabarti~\cite{Stoeckl2023,CS2024}.
In this problem, the input is a stream of elements from universe $[n]$, and at each point throughout a stream of length $m < n$, the algorithm
is asked to output an element that has not appeared in the stream so far.
Interestingly, MIF is a problem in  which memoryless deterministic adversaries do turn out to have a significant advantage over oblivious adversaries in some regimes.
For instance, the ``echo'' strategy in which the adversary uses the last output of the algorithm as the next stream element turns out very effective,
because it can be used to force the algorithm to find a new missing element.
However, MIF is not a problem for which multiplicative approximation has any meaningful interpretation.
Since many incompatible outputs are possible in this problem, and on each, the adversary can act differently,
this seems to fall into the topic considered in Question~\ref{q:large_range}.
As such, the results we obtain here are incompatible with the existing literature on MIF,
but MIF could be a good starting point for further exploration of
the power of memoryless and low--memory adversaries
in various types of streaming problems.

\bibliographystyle{alpha}
\bibliography{references}

\appendix

\section{The Power of Persistent Randomness}

\begin{lemma}\label{lem:persistent_randomness}
Consider an extension of our model in which the adversary has query access to a~persistent
infinite random binary sequence. The sequence is drawn before the interaction
with the algorithm and the adversary has query access to it throughout the entire
interaction. Each bit of the string is drawn from the uniform distribution on
$\set{0,1}$, independently of any other randomness drawn in the algorithm.

For any algorithm, if there is an adversary who can
break its estimates with probability greater than $\delta$ in this extended model,
then there is an adversary in our model (i.e., without the additional persistent randomness) who can break its estimates
with probability greater than $\delta$ as well.
\end{lemma}

\begin{proof}
Let $\delta' > \delta$ be the probability with which the adversary breaks the
estimates in the extended model. This
occurs for some fixed stream length known to both of the parties. By an
averaging argument, there must be a setting $r_\star \in \set{0,1}^{\mathbb N}$
of the infinite random string for which the adversary succeeds with probability
at least $\delta'$ as well. Hence there is an adversary with no access to an
infinite persistent random string but just to $r_\star$ who also succeeds with
probability at least $\delta'$.

If we could include $r_\star$ in the adversary's internal logic,
we would be able to show that this adversary can be implemented in
our model, which provides no persistent randomness.
This may be difficult to do achieve directly, however,
because including an arbitrary infinite string in the adversary's code is
impossible as there are uncountably many such strings. Fortunately, the
adversary does not have to store the entire $r_\star$ but just its finite
prefix. In every game between the adversary and algorithm, the adversary
accesses a finite number of entries in $r_\star$. Let $X$ be a random variable
equal to the highest index of a bit accessed throughout the game (or 0 if none
are accessed). $X$ is distributed on $\mathbb N$ and there is a threshold $t \in
\mathbb N$ such that $\Pr[X \ge t] < (\delta'-\delta)/2$. We modify the
adversary to remember the first $t$ bits of $r_\star$ and whenever accessing any
other bit of $r_\star$ is attempted, the adversary simply outputs a fixed update. Before the
modification, the adversary was able to make the algorithm output an incorrect
estimate with probability at least $\delta'$. The behavior of the modified
adversary, who only knows the first $t$ bits of $r_\star$, can diverge only if
any bit beyond the first $t$ bits is accessed and since this happens with
probability less than $(\delta' - \delta)/2$, the modified adversary is still
able to break the estimates of the algorithm with probability at least $\delta'
- (\delta' - \delta)/2 = (\delta' + \delta)/2 > \delta$.
\end{proof}

\section{Missing Proofs}

\subsection{Proof of Lemma~\ref{lem:net}.}
\begin{proof}[Proof of Lemma~\ref{lem:net}]
    Let $\mathcal{N} \eqdef \{(1+\eps)^i: i \in \mathbb Z_{+} \land
(1+\eps)^i \le \alpha\} \cup \{\alpha\}$. It is clear that $\mathcal{N} \subseteq [1, \alpha]$.  Additionally,
the size of $\mathcal {N}$ is bounded by $ \lceil \log_{1+\eps}\alpha \rceil
\le 1 + \log_{1+\eps}\alpha = O(\log_{1+\eps}\alpha) = O(\eps^{-1} \log
\alpha)$. now for any real $x \in [1,\alpha]$, take the smallest $i \in \mathbb Z_{+}$ such that $(1+\eps)^{i-1} \le x < (1+\eps)^i$. Then $(1+\eps)^i \le (1+\eps)x$, so if $(1+\eps)^i \in \mathcal{N}$, then this can serve as a valid choice of $y$ for $x$ as in the lemma statement. Otherwise, $(1+\eps)^i > \alpha$, in which case 
$x \le \alpha \le (1+\eps)^i \le (1+\eps)x$,  where $\alpha \in \mathcal{N}$
plays the role of $y$ in the definition.
\end{proof}

\subsection{Proof of Lemma~\ref{lem:amplification}.}

\begin{proof}[Proof of Lemma~\ref{lem:amplification}]
Consider any fixed stream and let $y_\star$ be the value of $f$ at the end of
the stream. By definition, $\alg$ outputs a $(1+\epsilon)$--multiplicative
approximation to $y_\star$ with probability at least $9/10$. In other words,
with probability at least $9/10$, it outputs an estimate $y$ such that $y\in I$,
where $I \eqdef [y_\star,(1+\eps)y_\star)$.

Suppose that we run $t$ independent copies of $\alg$.
For each $i \in [t]$, let $X_i$ be an indicator variable
equal to $1$ if the $i$-th copy of $\alg$ returns
an estimate in $I$, and 0 otherwise.
We have $\mu \eqdef \mathbb E[\sum_{i=1}^{t} X_i]
\ge 9t/10$. Since $X_i$'s are independent, we use the Chernoff bound
to bound the probability that at least half of all the estimates produced by the copies of $\alg$
are outside of $I$:
\begin{align*}
\Pr\left[\sum_{i=1}^{t} X_i \le t/2 \right]
&=
\Pr\left[\sum_{i=1}^{t} X_i \le \left(1-\frac{4}{9}\right) \cdot \frac{9}{10}t \right]
\le
\Pr\left[\sum_{i=1}^{t} X_i \le \left(1-\frac{4}{9}\right) \cdot \mu \right]\\
& \le
\exp\left(- \left(\frac{4}{9}\right)^2 \frac{1}{2} \mu\right)
\le
\exp\left(- \frac{8}{81} \cdot \frac{9}{10} t\right)
\le \exp(-\ln (1/\delta)) = \delta.
\end{align*}
Note that if the median of the estimates is outside of $I$, then at least $t/2$ of the estimates
are strictly lower than $y_\star$ or at least $t/2$ of them are greater than or equal to $(1+\eps)y_\star$.
Hence if the number of the estimates outside of $I$ is less than $t/2$,
the median of the estimates belongs to $I$.
We already know that this happens with probability at least $1 - \delta$.
\end{proof}

\end{document}

%% file: macros.tex
\newcommand{\Adv}{\textnormal{\texttt{Adversary}}}
\newcommand{\Alg}{\textnormal{\texttt{Algorithm}}}

\newcommand{\R}{\mathbb{R}}
\newcommand{\N}{\mathbb{N}}